\documentclass[aps,pra,twocolumn,10pt,nofootinbib]{revtex4-2}

\usepackage[utf8]{inputenc}
\usepackage[T1]{fontenc}
\usepackage{amsmath,amssymb,mathtools}
\usepackage{amsthm}
\usepackage[hidelinks]{hyperref}

\newtheorem{theorem}{Theorem}[section]

\newtheorem{proposition}{Proposition}

\begin{document}

\title{Exact series formulas for the capacities of the amplitude damping channel}

\author{Stefano Pirandola}
\affiliation{Department of Computer Science, University of York,
York YO10 5GH, United Kingdom}

\begin{abstract}
We derive explicit, absolutely convergent series formulas for the quantum, unassisted classical, and entanglement-assisted classical capacities of the qubit amplitude damping channel, eliminating residual optimizations and implicit roots. We also obtain analytical upper and lower bounds on two-way assisted quantum and private communication capacities by evaluating a balanced-squashing bound and optimizing the reverse coherent information. Our results characterize the optimal input populations, establish convergence and truncation properties, and provide analytical benchmarks for quantum and classical communication over dissipative channels.
\end{abstract}

\maketitle

\section{Introduction}

Determining the information-carrying capacities of noisy quantum
channels is a central problem in quantum Shannon theory. The
quantum capacity characterizes the maximum asymptotic rate of
reliable quantum communication
\cite{NielsenChuang2010,Simeone2026,Serafini2017,Weedbrook2012}
and is generally expressed in terms of the regularized coherent
information~\cite{Lloyd1997,Devetak2005}. Its evaluation
therefore involves both an optimization over input states and a
regularization over multiple channel uses. Although the latter
simplifies for certain channels, obtaining explicit expressions
without residual optimizations remains challenging. Similar
difficulties arise in determining classical communication
capacities and analytical bounds on assisted communication rates.

The qubit amplitude damping channel provides a fundamental model
of irreversible energy relaxation in a two-level quantum system.
It is also a useful setting in which to investigate these
different communication limits. In its degradable regime, the
quantum capacity reduces to the single-letter coherent
information~\cite{DevetakShor2005,WolfPerezGarcia2007}.
However, its standard expression retains a scalar optimization
over the input excited-state population
\cite{GiovannettiFazio2005,WolfPerezGarcia2007}.
Eliminating this optimization provides an explicit analytical
benchmark and enables the evaluation of related communication
bounds.

In this work, we derive an exact series formula for the quantum
capacity of the amplitude damping channel. We first transform
the stationarity condition into a scalar root equation, whose
unique positive solution determines both the capacity and the
optimal input population. Applying Lagrange-B\"urmann inversion
\cite{Gessel2016}, we obtain an explicit series containing
neither a residual optimization nor an implicitly defined root.
We establish absolute convergence at the physical parameter
value, including the boundary of the convergence disk, and
characterize the endpoint and threshold behavior.

We next exploit this solution to obtain analytical bounds on
two-way assisted quantum communication, entanglement
distribution, and secret-key generation. In particular, the
balanced amplitude damping squashing construction yields an
upper bound that coincides with the quantum capacity of an
amplitude damping channel with half the damping probability.
We complement this result by analytically optimizing the
reverse coherent information, obtaining an explicit lower
bound. Together, these expressions provide optimization-free
bounds on the corresponding assisted capacities.

We then extend the inversion approach to the unassisted
classical capacity $C$ and the entanglement-assisted classical
capacity $C_E$. Their respective single-letter
characterizations
\cite{TangZhuBaiWang2026,BennettShorSmolinThapliyal2002}
also involve scalar input optimizations. By introducing a
common normalized moment formalism, we derive exact series
for both capacities using the same coefficient polynomials.
We establish absolute convergence throughout the interior
of the physical damping range and obtain explicit expressions
for the optimal input populations and controlled truncation
errors. This common framework provides a unified analytical
treatment of the quantum and classical communication limits
of the amplitude damping channel.

The paper is organized as follows. Section~\ref{sec:AD}
reviews the amplitude damping channel, while
Sec.~\ref{sec:closed} presents the quantum-capacity series
and its limiting behavior. Section~\ref{sec:twoway}
establishes the two-way assisted capacity bounds.
Sections~\ref{capuni:sec-classical} and
\ref{capuni:sec-ea} present the classical and
entanglement-assisted classical capacity formulas,
respectively. Section~\ref{sec:conclusion} concludes.
The detailed derivations are collected in
Appendices~\ref{app:proof}, \ref{app:RCI},
and~\ref{capuni:app}.

\section{Amplitude damping channel and quantum capacity}
\label{sec:AD}

The qubit amplitude damping channel $\mathcal{A}_{\gamma}$ has
Kraus operators
\begin{equation}
E_0=
\begin{pmatrix}
1&0\\
0&\sqrt{1-\gamma}
\end{pmatrix},
\qquad
E_1=
\begin{pmatrix}
0&\sqrt{\gamma}\\
0&0
\end{pmatrix},
\label{eq:kraus}
\end{equation}
where $0\leq\gamma\leq1$ is the damping probability. Its action is
\begin{equation}
\mathcal{A}_{\gamma}(\rho)
=E_0\rho E_0^{\dagger}+E_1\rho E_1^{\dagger}.
\end{equation}
A Stinespring isometry is specified by
\begin{align}
|0\rangle&\longmapsto|0\rangle_B|0\rangle_E,
\nonumber\\
|1\rangle&\longmapsto
\sqrt{1-\gamma}\,|1\rangle_B|0\rangle_E
+\sqrt{\gamma}\,|0\rangle_B|1\rangle_E.
\label{eq:stinespring}
\end{align}
With these output bases, the complementary channel is
$\widetilde{\mathcal{A}}_{\gamma}=\mathcal{A}_{1-\gamma}$.

For $0\leq\gamma\leq1/2$, the channel is degradable, and its quantum
capacity is \cite{DevetakShor2005,WolfPerezGarcia2007}
\begin{equation}
Q(\mathcal{A}_{\gamma})
=\max_{\rho}I_c(\rho,\mathcal{A}_{\gamma}),
\label{eq:coherent}
\end{equation}
where
\begin{equation}
I_c(\rho,\mathcal{A}_{\gamma})
:=S[\mathcal{A}_{\gamma}(\rho)]
-S[\widetilde{\mathcal{A}}_{\gamma}(\rho)],
\end{equation}
with $S(\rho):=-\operatorname{Tr}(\rho\log_2\rho)$ being the
von Neumann entropy.

Coherent information is concave in the input for degradable
channels. Together with phase covariance, this implies that
dephasing the input in the energy basis cannot decrease the
objective. It is therefore sufficient to consider
\begin{equation}
\rho_u=(1-u)|0\rangle\langle0|+u|1\rangle\langle1|,
\qquad 0\leq u\leq1.
\label{eq:diagonal}
\end{equation}
For this input, the excitation probabilities at the output and in
the environment are $(1-\gamma)u$ and $\gamma u$, respectively.
Thus
\begin{equation}
I_c(u,\gamma)=h_2[(1-\gamma)u]-h_2(\gamma u),
\label{eq:Ic}
\end{equation}
where
\begin{equation}
h_2(x):=-x\log_2x-(1-x)\log_2(1-x)
\label{eq:binaryentropy}
\end{equation}
is the binary entropy, with continuous endpoint values
$h_2(0)=h_2(1)=0$. We also use the entropy in nats,
\begin{equation}
H(x):=(\ln2)h_2(x).
\label{eq:naturalentropy}
\end{equation}
The capacity becomes
\begin{equation}
Q(\mathcal{A}_{\gamma})
=\max_{0\leq u\leq1}
\big\{h_2[(1-\gamma)u]-h_2(\gamma u)\big\},
\label{eq:Qmax}
\end{equation}
for $0\leq\gamma\leq1/2$
\cite{GiovannettiFazio2005,WolfPerezGarcia2007}.

For $0<\gamma<1/2$ and $0<u<1$, the objective is strictly concave:
\begin{equation}
\frac{\partial^2 I_c(u,\gamma)}{\partial u^2}
=-\frac{1-2\gamma}
{\ln2\,u[1-(1-\gamma)u](1-\gamma u)}<0.
\label{eq:strictconcavity}
\end{equation}
It has a unique maximizing population in the interior of $[0,1]$. At $\gamma=0$, the
channel is the identity and the unique maximizing population is
$u=1/2$. At $\gamma=1/2$, the objective in Eq.~\eqref{eq:Qmax}
vanishes for every $u$. For $\gamma\geq1/2$, the channel is
antidegradable with zero quantum capacity~\cite{WolfPerezGarcia2007}.

\section{Explicit quantum-capacity formula}
\label{sec:closed}

For $0<\gamma<1/2$, define
\begin{equation}
k:=\frac{\gamma}{1-2\gamma},
\qquad
c_k:=\frac{k^k}{(k+1)^{k+1}}.
\label{eq:kckQ}
\end{equation}
The following result evaluates the maximization in
Eq.~\eqref{eq:Qmax}.

\begin{theorem}[Quantum-capacity series]
\label{thm:closedQ}
For $0<\gamma<1/2$, the quantum capacity of the amplitude damping
channel is
\begin{equation}
Q(\mathcal{A}_{\gamma})
=\frac{1}{\ln2}
\sum_{n=1}^{\infty}\frac{(-1)^{n-1}}{n}
\binom{(k+1)n-1}{n-1}c_k^n,
\label{eq:Qseries}
\end{equation}
where $k$ and $c_k$ are given by Eq.~\eqref{eq:kckQ}. This series
is absolutely convergent for every fixed $0<\gamma<1/2$.
\end{theorem}

For the generalized binomial coefficients appearing here,
\begin{equation}
\binom{x}{m}:=
\frac{\Gamma(x+1)}{\Gamma(m+1)\Gamma(x-m+1)},~~m=0,1,2,\ldots,
\end{equation}
where $\Gamma$ is the gamma function. All gamma-function arguments
in Eq.~\eqref{eq:Qseries} are positive. The first terms are
\begin{align}
Q(\mathcal{A}_{\gamma})=\frac{1}{\ln2}\bigg[
&c_k-\frac{2k+1}{2}c_k^2
\nonumber\\
&+\frac{(3k+2)(3k+1)}{6}c_k^3
\nonumber\\
&-\frac{(4k+3)(4k+2)(4k+1)}{24}c_k^4
+\cdots\bigg].
\label{eq:firstterms}
\end{align}
Equation~\eqref{eq:Qseries} is explicit in the damping probability
and contains no residual optimization or implicitly defined root.
The series converges algebraically, and Euler's transformation
provides a practical acceleration for numerical evaluation\footnote{
Writing $Q=\sum_{n=1}^{\infty}(-1)^{n-1}b_n$, where $b_n$ is
the positive term magnitude in Eq.~\eqref{eq:Qseries}, including
the factor $1/\ln2$,
Euler's transformation gives the equivalent representation
$Q=\sum_{m=0}^{\infty}2^{-m-1}\Delta^m b_1$.
Here $\Delta b_n:=b_n-b_{n+1}$ and $\Delta^0b_n:=b_n$.
The transformation preserves the infinite sum and accelerates
its numerical evaluation.}.

The proof of Theorem~\ref{thm:closedQ} is given in Appendix~\ref{app:proof}. It first identifies the unique
root $z_k>1$ of
\begin{equation}
(z_k-1)z_k^k=c_k
\label{eq:rootpreview}
\end{equation}
and establishes the equivalent root representation
\begin{equation}
Q(\mathcal{A}_{\gamma})=\log_2z_k.
\label{eq:rootcapacitypreview}
\end{equation}
For $0<\gamma<1/2$, the same root characterizes the unique
maximizing population:
\begin{equation}
u_*(\gamma)=
\frac{(2k+1)(z_k-1)}{1+(k+1)(z_k-1)}.
\label{eq:optimalpopulation}
\end{equation}
This formula also becomes explicit through the series by using
$z_k=2^{Q(\mathcal{A}_{\gamma})}$ with
Eq.~\eqref{eq:Qseries}. The series and root representations thus
give equivalent characterizations of the optimizing population.

The endpoint capacities are
\begin{equation}
Q(\mathcal{A}_0)=1,
\qquad Q(\mathcal{A}_{1/2})=0,
\label{eq:endpoints}
\end{equation}
and the interior formulas approach these values continuously. In
particular, $u_*(\gamma)\to1/2$ as $\gamma\to0^+$.
At the other endpoint, introduce
\begin{equation}
\omega:=W_0(e^{-1}),
\qquad \omega e^{\omega}=e^{-1},
\label{eq:omega}
\end{equation}
where $W_0$ is the principal real branch of the Lambert function.
Appendix~\ref{app:proof} shows that, for $\gamma\to\tfrac12^-$, we have 
\begin{align}
Q(\mathcal{A}_{\gamma})
&=\frac{2\omega}{\ln2}(1-2\gamma)
+o(1-2\gamma), \label{eq:thresholdQ} \\
u_*(\gamma)
&\rightarrow\frac{2\omega}{1+\omega}
\simeq0.4356234114.
\label{eq:thresholdpopulation}
\end{align}
The limiting population in Eq.~\eqref{eq:thresholdpopulation}
describes the approach from the interior. At the threshold itself,
every population maximizes the identically zero coherent-information
objective.

\section{Application to two-way assisted capacities}
\label{sec:twoway}

We now consider unlimited two-way classical assistance for the
physical channel $\mathcal{A}_{\gamma}$, with $0\leq\gamma\leq1$.
Let $Q_2$, $D_2$, and $K_2$ denote its two-way assisted quantum,
entanglement-distribution, and secret-key capacities, respectively.
Classical communication is public in the definition of $K_2$, and no
entanglement or secret key is initially shared\footnote{Note that $K_2$ is also often denoted by $K$, since secret-key
agreement typically allows unlimited two-way public classical
communication.}. These capacities
satisfy $Q_2=D_2\leq K_2$.

The squashed entanglement of a bipartite state is one half of the
infimum of the quantum conditional mutual information over
extensions of that state \cite{ChristandlWinter2004,TakeokaGuhaWilde2014}. Maximizing
the squashed entanglement between a reference system and the channel
output over purified channel inputs defines the channel quantity
$E_{\mathrm{sq}}$. The channel converse gives
\begin{align}
Q_2(\mathcal{A}_{\gamma})
=D_2(\mathcal{A}_{\gamma})
&\leq K_2(\mathcal{A}_{\gamma}) \leq E_{\mathrm{sq}}(\mathcal{A}_{\gamma}).
\label{eq:capacityhierarchy}
\end{align}
A fixed squashing channel acting on the environment supplies an
upper bound on $E_{\mathrm{sq}}$ by maximizing one half of the
corresponding conditional mutual information over the channel input.

For the balanced amplitude damping squashing construction, whose
squashing channel has damping probability $1/2$, the channel
squashed entanglement satisfies
$E_{\mathrm{sq}}(\mathcal{A}_{\gamma})\leq U(\gamma)$, where\footnote{Supplementary Eq.~(S231) of Ref.~\cite{PLOB2017} gives the correct
input-state maximization, which is equivalent to
Eq.~\eqref{eq:twowaymax} here (with a different notation for the
damping probability and the input population). An imprecision occurs
in the subsequent removal of this maximization: the objective is
stated there to be symmetric in the input population, so that the
maximizer is set to $1/2$. This leads to Supplementary Eq.~(S232) and
to Eq.~(49) of Ref.~\cite{PLOB2017}. In general, however, the
objective is not symmetric in the input population and its maximizer
need not be $1/2$. Theorem~\ref{thm:twowaybound} retains and evaluates
this maximization exactly, thereby correcting this simplification
while leaving the preceding balanced-squashing construction
unchanged.}
\begin{equation}
U(\gamma):=
\max_{0\leq u\leq1}
\left\{
h_2\!\left[\left(1-\frac{\gamma}{2}\right)u\right]
-h_2\!\left(\frac{\gamma u}{2}\right)
\right\}.
\label{eq:twowaymax}
\end{equation}

The balanced-squashing bound can be evaluated through the ordinary
quantum capacity of an effective amplitude damping channel. In fact, by comparing Eqs.~\eqref{eq:Qmax} and~\eqref{eq:twowaymax}, we note the
relevant identity
\begin{equation}
U(\gamma)=Q(\mathcal{A}_{\gamma/2}),
\label{eq:fundamentalidentity}
\end{equation}
which relates the two input optimizations throughout the physical
damping range. Together with the reverse-coherent-information
(RCI) achievable rate~\cite{GarciaPatron2009,Pirandola2009SecretKey}, evaluated in
Appendix~\ref{app:RCI}, this gives the following result.

\begin{theorem}[Series bounds for two-way assisted capacities]
\label{thm:twowaybound}
For every $0\leq\gamma\leq1$ and any two-way assisted capacity
$\mathcal{C}_2\in\{Q_2,D_2,K_2\}$,
\begin{equation}
Q(\mathcal{A}_{\gamma/(1+\gamma)})
\leq\mathcal{C}_2(\mathcal{A}_{\gamma})
\leq Q(\mathcal{A}_{\gamma/2}).
\label{eq:twowayQbound}
\end{equation}
For $0<\gamma<1$, the upper bound is given by the absolutely
convergent series in Theorem~\ref{thm:closedQ}, with
\begin{align}
k=\kappa_{\gamma}:=\frac{\gamma}{2(1-\gamma)},~~c_k=\frac{\kappa_{\gamma}^{\kappa_{\gamma}}}
{(\kappa_{\gamma}+1)^{\kappa_{\gamma}+1}}.
\label{eq:effectiveparameters}
\end{align}
The lower bound is the optimized RCI rate and is given by the same
series with $k=2\kappa_\gamma$ and $c_k$ as in
Eq.~\eqref{eq:kckQ}. Both bounds equal one at $\gamma=0$ and
zero at $\gamma=1$.
\end{theorem}

\noindent\textit{Proof.}
For every $\gamma\in[0,1]$, the effective damping probability
$\gamma/2$ belongs to $[0,1/2]$. Hence
$\mathcal{A}_{\gamma/2}$ is degradable, and its quantum capacity is
given by Eq.~\eqref{eq:Qmax}. Replacing the damping probability
in that equation by $\gamma/2$ gives precisely the maximization
defining $U(\gamma)$ in Eq.~\eqref{eq:twowaymax}, proving
Eq.~\eqref{eq:fundamentalidentity}. Combining this identity with
the capacity hierarchy in Eq.~\eqref{eq:capacityhierarchy} and the
bound $E_{\mathrm{sq}}(\mathcal{A}_{\gamma})\leq U(\gamma)$ stated
above establishes the upper bound in Eq.~\eqref{eq:twowayQbound}.
For $0<\gamma<1$,
Theorem~\ref{thm:closedQ} applies to the effective channel, and
substitution of $\gamma/2$ into Eq.~\eqref{eq:kckQ} yields
Eq.~\eqref{eq:effectiveparameters}. The lower bound and its series
evaluation are proved in Appendix~\ref{app:RCI}. The endpoint values
follow from Eq.~\eqref{eq:endpoints}.
\hfill$\square$

The identity in Eq.~\eqref{eq:fundamentalidentity} also determines
the optimizing input. For $0<\gamma<1$, the balanced-squashing
objective has the unique maximizing population
\begin{equation}
u_*^{\mathrm{sq}}(\gamma)=u_*(\gamma/2),
\label{eq:optimalpopulationtwoway}
\end{equation}
where $u_*$ is characterized by Eq.~\eqref{eq:optimalpopulation}.
At $\gamma=0$, the unique maximizing population is $u=1/2$.
At $\gamma=1$, the objective vanishes identically, so every
$u\in[0,1]$ is a maximizer.

The same correspondence transfers the threshold behavior of the
ordinary quantum capacity to the strong-damping behavior of the
bound, so Eq.~\eqref{eq:thresholdQ} gives
\begin{equation}
U(\gamma)=\frac{2\omega}{\ln2}(1-\gamma)+o(1-\gamma),~~\gamma\to1^-,
\label{eq:thresholdU}
\end{equation}
where $\omega=W_0(e^{-1})$ and
$2\omega/\ln2\simeq0.8034788298$.

The upper bound in Theorem~\ref{thm:twowaybound} evaluates the
input-state optimization for the specified balanced squashing channel.
The resulting quantity $U(\gamma)$ bounds both the channel squashed entanglement
and the two-way assisted capacities. Establishing
tightness of this bound or optimizing over more general squashing
channels requires a separate analysis.

\section{Classical capacity}
\label{capuni:sec-classical}

The classical-capacity characterization in Theorem~5.5 of
Ref.~\cite{TangZhuBaiWang2026} admits the following explicit series.
For $0<\gamma<1$, write $a=1-\gamma$, $b=\gamma a$, and set
\begin{align}
d&=\sqrt{1-b},\qquad \ell=\ln\frac{1+d}{1-d},
\label{capuni:classical-auxiliary}\\
\delta_{\mathrm C}
&=\frac{b\ell}{d}-a\ln\frac{1+\gamma}{a}.
\label{capuni:classical-explicit-parameter}
\end{align}
Define the elementary functions
\begin{align}
D(z)&=\sqrt{1-b(1-z)^2},\notag\\
\Phi_{\mathrm C}(z)
&=a\ln\frac{1+\gamma+az}{(1+\gamma)(1-z)}
  +\frac{b\ell}{d}\notag\\
&\quad-\frac{b(1-z)}{D(z)}\ln\frac{1+D(z)}{1-D(z)},
\label{capuni:classical-explicit-kernel}
\end{align}
and the coefficients
\begin{equation}
c_n=\frac1n[z^{n-1}]
 \left(\frac{z}{\Phi_{\mathrm C}(z)}\right)^n,
\qquad n\ge1.
\label{capuni:classical-explicit-coefficients}
\end{equation}
Here $[z^j]$ denotes Taylor-coefficient extraction at $z=0$\footnote{
For a function $F$ analytic at the origin, $[z^j]F(z)$ denotes
the coefficient of $z^j$ in its Taylor expansion, equivalently
$F^{(j)}(0)/j!$. Thus, $c_n$ is obtained by expanding
$(z/\Phi_{\mathrm C}(z))^n$, selecting the coefficient of $z^{n-1}$,
and dividing by $n$.};
the quotient is defined there by continuity.

\begin{theorem}[Classical-capacity series]
\label{capuni:thm-classical}
For every $0<\gamma<1$,
\begin{align}
C(\mathcal A_\gamma)
&=h_2\!\left(\frac a2\right)
 -h_2\!\left(\frac{1-d}{2}\right)\notag\\
&\quad+\frac1{2\ln2}\sum_{n=1}^{\infty}
 \frac{c_n\delta_{\mathrm C}^{n+1}}{n+1}.
\label{capuni:classical-series}
\end{align}
The series converges absolutely. The endpoint values are
$C(\mathcal A_0)=1$ and $C(\mathcal A_1)=0$.
\end{theorem}

\section{Entanglement-assisted classical capacity}
\label{capuni:sec-ea}

For entanglement-assisted communication
\cite{BennettShorSmolinThapliyal2002}, retain $a=1-\gamma$ and set
\begin{align}
L&=-\ln a,\qquad x=\max\{1,L-\ln(1+L)\},
\label{capuni:ea-reference}\\
\delta_{\mathrm E}
&=2a\ln x-a\ln\frac{1+\gamma x}{a}\notag\\
&\quad+\gamma\ln\frac{1+ax}{\gamma}.
\label{capuni:ea-explicit-parameter}
\end{align}
Using the same coefficient-extraction convention, define
\begin{align}
\Phi_{\mathrm E}(z)
&=\ln(1+z)+a\ln\!\left(1+\frac{az}{1+\gamma x}\right)
\notag\\
&\quad-\gamma\ln\!\left(1+\frac{\gamma z}{1+ax}\right)
 -2a\ln\!\left(1-\frac zx\right),
\label{capuni:ea-explicit-kernel}\\
e_n&=\frac1n[z^{n-1}]
 \left(\frac{z}{\Phi_{\mathrm E}(z)}\right)^n,
\qquad n\ge1.
\label{capuni:ea-explicit-coefficients}
\end{align}

\begin{theorem}[Entanglement-assisted capacity series]
\label{capuni:thm-ea}
For every $0<\gamma<1$,
\begin{align}
C_E(\mathcal A_\gamma)
&=h_2\!\left(\frac{x}{1+x}\right)
 +h_2\!\left(\frac{ax}{1+x}\right)\notag\\
&\quad-h_2\!\left(\frac{\gamma x}{1+x}\right)\notag\\
&\quad+\frac1{(1+x)\ln2}\sum_{n=1}^{\infty}
 \frac{e_n\delta_{\mathrm E}^{n+1}}{n+1}.
\label{capuni:ea-series}
\end{align}
The series converges absolutely. Moreover,
$C_E(\mathcal A_0)=2$, $C_E(\mathcal A_{1/2})=1$,
and $C_E(\mathcal A_1)=0$.
\end{theorem}

Note that the entanglement-assisted quantum capacity satisfies
$Q_E(\mathcal A_\gamma)=\tfrac12 C_E(\mathcal A_\gamma)$
with unlimited preshared entanglement
\cite{BennettShorSmolinThapliyal2002}.
Thus, Theorem~\ref{capuni:thm-ea} also provides an exact
series for $Q_E$.

Appendix~\ref{capuni:app} gives the common normalized formalism,
proofs, optimal populations, finite coefficient formulas,
and truncation bounds.

\section{Conclusion}
\label{sec:conclusion}

We have evaluated the remaining scalar optimization in the quantum
capacity of the qubit amplitude damping channel.
Theorem~\ref{thm:closedQ} gives the capacity in the interior of the
degradable regime as the explicit, absolutely convergent series
in Eq.~\eqref{eq:Qseries}. The proof also establishes the root
representation, the unique optimal input population, the endpoint
limits, and the leading behavior at the
degradable-antidegradable threshold.

The same solution evaluates the balanced-squashing upper bound for
two-way assisted quantum communication, entanglement distribution,
and secret-key generation. The identity
$U(\gamma)=Q(\mathcal A_{\gamma/2})$ relates this bound to the
ordinary quantum capacity at half the physical damping probability.
Together with the optimized RCI lower bound, it gives the explicit
bounds in Theorem~\ref{thm:twowaybound} across the full damping range.
The upper bound concerns a fixed squashing construction;
determining its global optimality requires a separate analysis.

We have also extended the inversion method to the unassisted
classical and entanglement-assisted classical capacities.
Theorems~\ref{capuni:thm-classical} and~\ref{capuni:thm-ea} evaluate
both capacities through the explicit series in
Eqs.~\eqref{capuni:classical-series} and~\eqref{capuni:ea-series},
using a common coefficient construction.
Their positive moment kernels yield absolute convergence for every
$0<\gamma<1$, explicit maximizing populations, and common truncation
bounds. Only the reference population and the coefficient data change
between the two communication tasks.

Together, these results provide explicit analytical benchmarks for
three communication capacities of a basic dissipative channel and
for the comparison of achievable rates and converse bounds in
assisted quantum communication. The same approach may also be useful
in identifying other channel models whose remaining scalar capacity
optimizations admit convergent series representations.

\begin{acknowledgments}

We acknowledge support from UKRI via the Integrated Quantum Networks
Research Hub (IQN, EP/Z533208/1).

\end{acknowledgments}

\section*{AI statement}

OpenAI's ChatGPT models (GPT-5.6 Sol and GPT-6 Astra) were used during the preparation of this manuscript to generate candidate mathematical derivations and proofs, assist with their development, and improve the presentation of the manuscript. All mathematical results and proofs were independently examined and verified by the author, who assumes full responsibility for the accuracy, originality, and scientific content of the work.

\appendix
\section{Proof of Theorem~\ref{thm:closedQ}}
\label{app:proof}

Fix $0<\gamma<1/2$. The objective in Eq.~\eqref{eq:Qmax} is
continuous on $[0,1]$, vanishes at both endpoints, and is strictly
concave by Eq.~\eqref{eq:strictconcavity}. It is therefore positive
in the interior and has a unique interior maximizer. In particular,
this maximizer satisfies the stationarity condition used below.

Set
\begin{equation}
d:=1-2\gamma,~~ k:=\frac{\gamma}{d}.
\label{eq:appendixdk}
\end{equation}
Then
\begin{equation}
\gamma=kd,~~ 1-\gamma=(k+1)d,
~~ d=\frac1{2k+1}.
\label{eq:appendixrelations}
\end{equation}
The rescaled population
\begin{equation}
t:=du,~~ 0\leq t\leq\frac1{2k+1},
\label{eq:def_t}
\end{equation}
puts the objective in the form
\begin{equation}
F(t):=H[(k+1)t]-H(kt),
\label{eq:Ft}
\end{equation}
where $H$ is defined in Eq.~\eqref{eq:naturalentropy}.
Thus $F(t)=\ln2\,I_c(u,\gamma)$.

Let $t_*$ denote the unique maximizing value. Differentiation gives
\begin{equation}
(k+1)\ln\frac{1-(k+1)t_*}{(k+1)t_*}
=k\ln\frac{1-kt_*}{kt_*}.
\label{eq:stationaryt}
\end{equation}
Using
\begin{equation}
H(x)=x\ln\frac{1-x}{x}-\ln(1-x),
\label{eq:Hidentity}
\end{equation}
the terms proportional to $t_*$ cancel by
Eq.~\eqref{eq:stationaryt}, leaving
\begin{equation}
F(t_*)=\ln\frac{1-kt_*}{1-(k+1)t_*}.
\label{eq:Fstar}
\end{equation}
Both the numerator and the denominator are positive in the interior
of the physical domain. Define
\begin{equation}
z:=\frac{1-kt_*}{1-(k+1)t_*}>1.
\label{eq:defzappendix}
\end{equation}
Solving for $t_*$ yields
\begin{equation}
t_* =\frac{z-1}{1+(k+1)(z-1)}.
\label{eq:tzappendix}
\end{equation}
The two ratios in the stationarity condition become
\begin{align}
\frac{1-(k+1)t_*}{(k+1)t_*}
&=\frac1{(k+1)(z-1)},
\nonumber\\
\frac{1-kt_*}{kt_*}&=\frac{z}{k(z-1)}.
\end{align}
Substitution into Eq.~\eqref{eq:stationaryt} gives
\begin{equation}
(z-1)z^k=\frac{k^k}{(k+1)^{k+1}}=c_k.
\label{eq:zimplicitappendix}
\end{equation}

For $z>1$, the function $g_k(z):=(z-1)z^k$ has derivative
\begin{equation}
g_k'(z)=z^{k-1}[(k+1)z-k]>0.
\end{equation}
It increases continuously from zero to infinity, so
Eq.~\eqref{eq:zimplicitappendix} has exactly one root $z_k>1$.
Equations~\eqref{eq:Fstar} and \eqref{eq:defzappendix} therefore
establish
\begin{equation}
Q(\mathcal{A}_{\gamma})=\frac{F(t_*)}{\ln2}=\log_2z_k.
\label{eq:Qrootappendix}
\end{equation}
Since $u_*=(2k+1)t_*$, Eq.~\eqref{eq:tzappendix} also gives
Eq.~\eqref{eq:optimalpopulation}.

The population obtained from the positive root is indeed physical.
Writing $w_k:=z_k-1$, Eq.~\eqref{eq:zimplicitappendix} implies
$0<w_k<c_k$. Moreover,
\begin{equation}
kc_k=\left(\frac{k}{k+1}\right)^{k+1}<1.
\end{equation}
Thus $kw_k<1$, which is precisely the condition ensuring
$(2k+1)w_k/[1+(k+1)w_k]<1$; positivity is immediate.

To obtain the series, rewrite the root equation as
\begin{equation}
w_k=c_k(1+w_k)^{-k}
\label{eq:wimplicitappendix}
\end{equation}
and consider the local analytic solution of
\begin{equation}
w=x(1+w)^{-k},~~w(0)=0.
\label{eq:wgeneral}
\end{equation}
The branch of $(1+w)^{-k}$ used here is analytic near $w=0$ and
positive for real $w>-1$. For an analytic function $G$, the
Lagrange-B\"urmann formula gives, for sufficiently small $|x|$,
\cite{Gessel2016}
\begin{equation}
G[w(x)]=G(0)+\sum_{n=1}^{\infty}
\frac{x^n}{n}\,[s^{n-1}]
\bigl(G'(s)(1+s)^{-kn}\bigr),
\label{eq:LBgeneral}
\end{equation}
where $[s^{n-1}]$ denotes coefficient extraction. Taking
$G(s)=\ln(1+s)$ gives
\begin{align}
\ln[1+w(x)]
&=\sum_{n=1}^{\infty}\frac{x^n}{n}
[s^{n-1}](1+s)^{-kn-1}
\nonumber\\
&=\sum_{n=1}^{\infty}\frac{(-1)^{n-1}}{n}
\binom{(k+1)n-1}{n-1}x^n.
\label{eq:LBseriesappendix}
\end{align}

The evaluation at $x=c_k$ requires a boundary argument. For fixed
$k>0$, the absolute value of the $n$th term at this point is
\begin{align}
a_n
:=\frac1n\binom{(k+1)n-1}{n-1}c_k^n
=\frac{\Gamma((k+1)n)}
{\Gamma(n+1)\Gamma(kn+1)}c_k^n.
\end{align}
Stirling's formula gives
\begin{equation}
a_n\sim\frac1{\sqrt{2\pi k(k+1)}}n^{-3/2},
\qquad n\to\infty.
\label{eq:seriesasymptotic}
\end{equation}
Consequently, the power series on the right-hand side of
Eq.~\eqref{eq:LBseriesappendix} has radius of convergence $c_k$
and converges absolutely at $x=c_k$.

For real $x\geq0$, the map $x=w(1+w)^k$ has a unique nonnegative
inverse, because
\begin{equation}
\frac{d}{dw}\bigl[w(1+w)^k\bigr]
=(1+w)^{k-1}[1+(k+1)w]>0
\end{equation}
for $w\geq0$. Its inverse is analytic in a neighborhood of every
point on the nonnegative real axis. The identity in
Eq.~\eqref{eq:LBseriesappendix}, established locally near $x=0$,
therefore extends along the positive axis throughout $0\leq x<c_k$.
Since the series is absolutely
convergent at $c_k$, it is uniformly convergent on $[0,c_k]$ by
comparison with $\sum_n a_n$. Taking $x\to c_k^-$ in
Eq.~\eqref{eq:LBseriesappendix} and using continuity of the
positive inverse gives
\begin{equation}
\ln(1+w_k)=\sum_{n=1}^{\infty}\frac{(-1)^{n-1}}{n}
\binom{(k+1)n-1}{n-1}c_k^n.
\end{equation}
Division by $\ln2$ proves Eq.~\eqref{eq:Qseries}, including its
absolute convergence, and completes the proof of
Theorem~\ref{thm:closedQ}.

We finally establish the endpoint statements. As $\gamma\to0^+$,
one has $k\to0^+$ and $c_k\to1$. The positive solution of
$w_k(1+w_k)^k=c_k$ satisfies $0<w_k<c_k<1$. Hence
$(1+w_k)^k\to1$ and $w_k\to1$, giving
$Q(\mathcal{A}_{\gamma})\to1$ and $u_*(\gamma)\to1/2$.
The absolute-convergence assertion in the theorem is restricted to
$k>0$; the formal $k=0$ limit of the series is the conditionally
convergent alternating harmonic series.

As $\gamma\to1/2^-$, one has $k\to\infty$. Define $y_k:=kw_k$.
The root equation becomes
\begin{equation}
y_k\left(1+\frac{y_k}{k}\right)^k
=\left(\frac{k}{k+1}\right)^{k+1}.
\label{eq:yimplicit}
\end{equation}
Since $0<y_k<1$, the family $y_k$ is bounded. Any accumulation
point $y$ as $k\to\infty$ satisfies $ye^y=e^{-1}$ by
Eq.~\eqref{eq:yimplicit}. The unique nonnegative solution is
$\omega=W_0(e^{-1})$, so $y_k\to\omega$. Therefore
\begin{align}
\frac{Q(\mathcal{A}_{\gamma})}{1-2\gamma}
&=\frac{(2k+1)\ln(1+y_k/k)}{\ln2}
\longrightarrow\frac{2\omega}{\ln2},
\nonumber\\
u_*(\gamma)
&=\frac{(2+1/k)y_k}{1+(1+1/k)y_k}
\longrightarrow\frac{2\omega}{1+\omega}.
\end{align}
These limits establish Eqs.~\eqref{eq:thresholdQ} and
\eqref{eq:thresholdpopulation}. The exact endpoint values follow
also directly from Eq.~\eqref{eq:Qmax}.

\section{Optimization-free reverse-coherent-information lower bound}
\label{app:RCI}

The reverse coherent information (RCI) is
$I_{\mathrm R}(\rho,\mathcal A_\gamma)
=S(\rho)-S[\widetilde{\mathcal A}_\gamma(\rho)]$.
For the Stinespring isometry $V_\gamma:A\to BE$ in
Eq.~\eqref{eq:stinespring}, this equals
$S(B|E)_{V_\gamma\rho V_\gamma^\dagger}$ and is therefore concave
in $\rho$. Phase covariance then allows the optimization to be
restricted to the diagonal inputs $\rho_u$ of
Eq.~\eqref{eq:diagonal}.
The resulting rate is achievable by entanglement distillation with
backward classical communication~\cite{GarciaPatron2009} and hence
lower bounds the two-way assisted quantum capacity:
\begin{equation}
Q_2(\mathcal A_\gamma)
\geq
R_{\mathrm{RCI}}(\gamma)
:=
\max_{0\leq u\leq1}
\left\{h_2(u)-h_2(\gamma u)\right\}.
\label{eq:RCImax}
\end{equation}
This maximization can be eliminated analytically.

\begin{proposition}
\label{prop:RCIclosed}
For $0<\gamma<1$, define
\begin{equation}
c_\gamma=(1-\gamma)\gamma^{\gamma/(1-\gamma)},
\label{eq:RCIcgamma}
\end{equation}
and let $z_\gamma\in(0,1)$ be the unique solution of
\begin{equation}
z_\gamma
=
c_\gamma(1-z_\gamma)^{1/(1-\gamma)}.
\label{eq:RCIz}
\end{equation}
Then
\begin{equation}
R_{\mathrm{RCI}}(\gamma)
=
-\log_2(1-z_\gamma),
\label{eq:RCIclosed}
\end{equation}
and the maximizing excited-state population is
\begin{equation}
u_\gamma^\star
=
\frac{z_\gamma}{1-\gamma+\gamma z_\gamma}.
\label{eq:RCIustar}
\end{equation}
The endpoint values are
\begin{equation}
R_{\mathrm{RCI}}(0)=1,
\qquad
R_{\mathrm{RCI}}(1)=0.
\label{eq:RCIendpoints}
\end{equation}
\end{proposition}

\begin{proof}
Set $f_\gamma(u)=h_2(u)-h_2(\gamma u)$. For $0<\gamma<1$,
\begin{equation}
f_\gamma''(u)
=
-\frac{1-\gamma}
{\ln 2\,u(1-u)(1-\gamma u)}
<0,
\qquad 0<u<1.
\label{eq:RCIconcavity}
\end{equation}
Together with $f_\gamma'(0^+)=+\infty$ and
$f_\gamma'(1^-)=-\infty$, this establishes a unique interior
maximizer, satisfying
\begin{equation}
\ln\frac{1-u}{u}
=
\gamma\ln\frac{1-\gamma u}{\gamma u}.
\label{eq:RCIstationary}
\end{equation}
Introduce
\begin{equation}
z=\frac{(1-\gamma)u}{1-\gamma u},
\qquad
u=\frac{z}{1-\gamma+\gamma z}.
\label{eq:RCIzdef}
\end{equation}
Since
\begin{align}
1-u&=\frac{(1-\gamma)(1-z)}{1-\gamma+\gamma z},
\nonumber\\
1-\gamma u&=\frac{1-\gamma}{1-\gamma+\gamma z},
\end{align}
the stationarity condition becomes
\begin{equation}
z^{1-\gamma}
=
(1-\gamma)^{1-\gamma}\gamma^\gamma(1-z),
\end{equation}
which is equivalent to Eq.~(\ref{eq:RCIz}). Its left-hand side
$z$ increases strictly from zero to one, whereas its right-hand side
$c_\gamma(1-z)^{1/(1-\gamma)}$ decreases strictly from
$c_\gamma$ to zero. Hence $z_\gamma$ exists and is unique.

The identity
\begin{equation}
f_\gamma(u)
=
\log_2\frac{1-\gamma u}{1-u}
+
u f_\gamma'(u)
\end{equation}
and the relation $1-z=(1-u)/(1-\gamma u)$ give
$f_\gamma(u_\gamma^\star)=-\log_2(1-z_\gamma)$.
Finally, $f_0(u)=h_2(u)$ and $f_1(u)=0$ yield the endpoint values.
\end{proof}

The rate also admits the explicit series representation
\begin{equation}
R_{\mathrm{RCI}}(\gamma)
=
\frac{1}{\ln2}
\sum_{n=1}^{\infty}
\frac{(-1)^{n-1}}{n}
\binom{n/(1-\gamma)-1}{n-1}c_\gamma^n,
\label{eq:RCIseries}
\end{equation}
which is absolutely convergent for every $0<\gamma<1$.

To justify Eq.~(\ref{eq:RCIseries}), fix
$\alpha=(1-\gamma)^{-1}>1$ and consider the branch $z(t)$ of
$z=t(1-z)^\alpha$ analytic at $t=0$, with $z(0)=0$.
Lagrange-B\"urmann inversion gives, initially near the origin,
\begin{equation}
-\ln[1-z(t)]
=
\sum_{n=1}^{\infty}
\frac{(-1)^{n-1}}{n}
\binom{\alpha n-1}{n-1}t^n.
\label{eq:RCILagrange}
\end{equation}
The identity
\begin{equation}
\binom{\alpha n-1}{n-1}
=\frac{\Gamma(\alpha n)}
{\Gamma(n)\Gamma((\alpha-1)n+1)}
\end{equation}
and Stirling's formula give the convergence radius
\begin{equation}
r_\alpha
=
\frac{(\alpha-1)^{\alpha-1}}{\alpha^\alpha}
=
c_\gamma,
\end{equation}
and, at the physical evaluation point $t=c_\gamma$,
\begin{equation}
\frac{c_\gamma^n}{n}
\binom{\alpha n-1}{n-1}
\sim
\frac{n^{-3/2}}{\sqrt{2\pi\alpha(\alpha-1)}}.
\label{eq:RCIasymptotic}
\end{equation}
Thus the series converges absolutely at this boundary point.
The map $t=z/(1-z)^\alpha$ has strictly positive derivative for
$0\leq z<1$, so its nonnegative inverse is analytic and the local
identity in Eq.~(\ref{eq:RCILagrange}) extends throughout
$0\leq t<c_\gamma$. Continuity, with
$z(c_\gamma)=z_\gamma$, and Abel's theorem extend this identity
to $t=c_\gamma$. Division by $\ln2$ proves
Eq.~(\ref{eq:RCIseries}).

To identify this rate with the lower bound in
Theorem~\ref{thm:twowaybound}, substitute the effective damping
probability $\gamma/(1+\gamma)$ into Eq.~\eqref{eq:kckQ}.
For $0<\gamma<1$, this gives
$k=\gamma/(1-\gamma)=2\kappa_\gamma$ and $c_k=c_\gamma$.
The series in Eqs.~\eqref{eq:Qseries} and \eqref{eq:RCIseries}
then coincide, proving
\begin{equation}
R_{\mathrm{RCI}}(\gamma)
=
Q(\mathcal A_{\gamma/(1+\gamma)}).
\label{eq:RCIcapacityidentity}
\end{equation}
The identity extends to $\gamma=0,1$ by the endpoint values in
Eqs.~\eqref{eq:endpoints} and \eqref{eq:RCIendpoints}.
Combining Eq.~\eqref{eq:RCImax} with the capacity hierarchy in
Eq.~\eqref{eq:capacityhierarchy} proves the lower bound in
Eq.~\eqref{eq:twowayQbound}.

As a check, when $\gamma=1/2$, Eq.~(\ref{eq:RCIz}) reduces to
\begin{equation}
z_{1/2}=\frac14(1-z_{1/2})^2,
\end{equation}
whose unique solution in $(0,1)$ is
$z_{1/2}=3-2\sqrt2$. Consequently,
\begin{align}
u_{1/2}^\star&=1-\frac1{\sqrt2},
\nonumber\\
R_{\mathrm{RCI}}\!\left(\frac12\right)
&=-\log_2\!\left[2(\sqrt2-1)\right].
\end{align}

\section{Compact capacity formulas and convergence}
\label{capuni:app}

\subsection{Normalized formulas and correspondence}
\label{capuni:app-compact}

Retain $a=1-\gamma$, $b=\gamma a$, and $H=(\ln2)h_2$. Define
\begin{align}
f_{\mathrm{cl}}(u)
&=H(au)-H\!\left(\frac{1-\sqrt{1-4bu^2}}2\right),
\label{capuni:classical-f}\\
f_{\mathrm{ea}}(u)&=H(u)+H(au)-H(\gamma u).
\label{capuni:ea-objective}
\end{align}
The capacity characterizations cited in Secs.~\ref{capuni:sec-classical}
and~\ref{capuni:sec-ea} give
$\mathcal C_\nu=(\ln2)^{-1}\max_{0\le u\le1}f_\nu(u)$,
where $\mathcal C_{\mathrm{cl}}=C(\mathcal A_\gamma)$ and
$\mathcal C_{\mathrm{ea}}=C_E(\mathcal A_\gamma)$.
For the EA case, phase covariance and concavity of the channel
mutual information permit restriction to the diagonal inputs $\rho_u$.
For $x$ in Eq.~\eqref{capuni:ea-reference}, choose
\begin{align}
c_{\mathrm{cl}}&=r_{\mathrm{cl}}=\frac12,\notag\\
c_{\mathrm{ea}}&=\frac{x}{1+x},\qquad
r_{\mathrm{ea}}=\frac1{1+x}.
\label{capuni:reference-data}
\end{align}
For either $\nu\in\{\mathrm{cl},\mathrm{ea}\}$, set
\begin{align}
\kappa_\nu&=-f_\nu''(c_\nu),\qquad
t_\nu=-\frac{f_\nu'(c_\nu)}{\kappa_\nu r_\nu},
\label{capuni:normalization}\\
B_\nu(z)
&=\frac{f_\nu'(c_\nu-r_\nu z)-f_\nu'(c_\nu)}
        {\kappa_\nu r_\nu z},
\label{capuni:kernel}\\
A_{\nu,n}&=\frac1n[z^{n-1}]B_\nu(z)^{-n},\qquad n\ge1,
\label{capuni:series-coefficients}
\end{align}
with $B_\nu(0)=1$ by continuity. The compact capacity formulas are
\begin{align}
C(\mathcal A_\gamma)
&=\frac{f_{\mathrm{cl}}(1/2)}{\ln2}\notag\\
&\quad+\frac{\kappa_{\mathrm{cl}}}{4\ln2}
 \sum_{n=1}^{\infty}\frac{A_{\mathrm{cl},n}}{n+1}
 t_{\mathrm{cl}}^{n+1},
\label{capuni:classical-compact}\\
C_E(\mathcal A_\gamma)
&=\frac{f_{\mathrm{ea}}(c_{\mathrm{ea}})}{\ln2}\notag\\
&\quad+\frac{\kappa_{\mathrm{ea}}r_{\mathrm{ea}}^2}{\ln2}
 \sum_{n=1}^{\infty}\frac{A_{\mathrm{ea},n}}{n+1}
 t_{\mathrm{ea}}^{n+1}.
\label{capuni:ea-compact}
\end{align}
Both converge absolutely for $0<\gamma<1$.

These are precisely Eqs.~\eqref{capuni:classical-series}
and~\eqref{capuni:ea-series} in normalized notation. Indeed,
direct substitution gives
\begin{align}
\Phi_{\mathrm C}(z)&=\frac{\kappa_{\mathrm{cl}}}{2}
 zB_{\mathrm{cl}}(z),\notag\\
\Phi_{\mathrm E}(z)&=\kappa_{\mathrm{ea}}r_{\mathrm{ea}}
 zB_{\mathrm{ea}}(z),
\label{capuni:kernel-correspondence}
\end{align}
and hence
\begin{align}
\delta_{\mathrm C}&=\frac{\kappa_{\mathrm{cl}}}{2}t_{\mathrm{cl}},
&c_n&=\left(\frac2{\kappa_{\mathrm{cl}}}\right)^n A_{\mathrm{cl},n},
\notag\\
\delta_{\mathrm E}&=\kappa_{\mathrm{ea}}r_{\mathrm{ea}}t_{\mathrm{ea}},
&e_n&=\frac{A_{\mathrm{ea},n}}
 {(\kappa_{\mathrm{ea}}r_{\mathrm{ea}})^n}.
\label{capuni:coefficient-correspondence}
\end{align}

\subsection{Common convergence argument}
\label{capuni:app-inversion}

We verify below that both objectives are strictly concave,
$|t_\nu|<1/2$, and
\begin{equation}
B_\nu(z)=\int_{-1}^{1}\frac{d\mu_\nu(s)}{1-sz},
\qquad |z|<1,
\label{capuni:moment-form}
\end{equation}
for a probability measure $\mu_\nu$.
Suppressing $\nu$, this representation implies, on $|z|=\rho<1$,
\begin{equation}
\operatorname{Re}B(z)\ge\frac1{1+\rho},\qquad
|zB(z)|\ge\frac{\rho}{1+\rho}.
\label{capuni:positive-real}
\end{equation}
Thus $B$ has no zeros in the unit disk. For $|t|<1/2$, choose $|t|/(1-|t|)<\rho<1$.
On $|z|=\rho$, one has $|t|<|zB(z)|$, so
Rouch\'e's theorem gives exactly one zero, counted
with multiplicity, of $zB(z)-t$ in $|z|<\rho$.
Letting $\rho$ increase to one establishes uniqueness
in the unit disk. The zero is simple, so the holomorphic implicit function
theorem and uniqueness give a holomorphic inverse $Z(t)$
for $|t|<1/2$, with $|Z(t)|<1$.
Lagrange-B\"urmann inversion and Cauchy's estimate yield
\begin{equation}
Z(t)=\sum_{n\ge1}A_nt^n,\qquad |A_n|\le2^n.
\label{capuni:inverse-series}
\end{equation}
For real $t_\nu$, the inverse is real. Since
$r_\nu\le\min\{c_\nu,1-c_\nu\}$, the stationarity condition and
strict concavity give the unique maximizing population
\begin{equation}
u_\nu=c_\nu-r_\nu\sum_{n\ge1}A_{\nu,n}t_\nu^n.
\label{capuni:common-population}
\end{equation}
Integration by parts using $Z(t)B(Z(t))=t$ gives
\begin{equation}
f_\nu(u_\nu)-f_\nu(c_\nu)
=\kappa_\nu r_\nu^2\int_0^{t_\nu}Z_\nu(t)\,dt.
\label{capuni:integrated-inverse}
\end{equation}
Termwise integration proves
Eqs.~\eqref{capuni:classical-compact} and~\eqref{capuni:ea-compact}.

Let $\mathcal C_{\nu,N}$ and $u_{\nu,N}$ retain the first $N$
terms of the corresponding series. With $\eta_\nu=2|t_\nu|<1$,
the coefficient bound gives
\begin{align}
|\mathcal C_\nu-\mathcal C_{\nu,N}|
&\le\frac{\kappa_\nu r_\nu^2|t_\nu|}{(\ln2)(N+2)}
 \frac{\eta_\nu^{N+1}}{1-\eta_\nu},
\label{capuni:capacity-remainder}\\
|u_\nu-u_{\nu,N}|
&\le r_\nu\frac{\eta_\nu^{N+1}}{1-\eta_\nu}.
\label{capuni:population-remainder}
\end{align}

\subsection{Verification for the two capacities}
\label{capuni:app-verification}

It is convenient to use
\begin{equation}
K_\nu(w)=\frac{f_\nu'(c_\nu-w)-f_\nu'(c_\nu)}{w}
=\kappa_\nu B_\nu(w/r_\nu).
\label{capuni:raw-kernel}
\end{equation}

\paragraph{Classical capacity.}
Put $v=2\sqrt b$. The derivative of the output entropy is
\begin{equation}
s(u)=v^2u\int_0^\infty\frac{dy}{1+vu\cosh y}.
\label{capuni:classical-s}
\end{equation}
Hence $s'(u)>0$ and
$f_{\mathrm{cl}}''(u)=-a/[u(1-au)]-s'(u)<0$.
For $c=1/2$, $\alpha=2a/(1+\gamma)$, and
$\xi(y)=v\cosh y/(1+vc\cosh y)$, direct subtraction gives
\begin{align}
K_{\mathrm{cl}}(w)
&=a\int_{-\alpha}^{2}\frac{d\xi}{1-\xi w}\notag\\
&\quad+v^2\int_0^\infty
 \frac{dy}{(1+vc\cosh y)^2[1-\xi(y)w]}.
\label{capuni:classical-raw-kernel}
\end{align}
Both measures are positive, $\alpha<2$, and $0<\xi(y)<2$.
Their total mass is $\kappa_{\mathrm{cl}}$; rescaling the support
by $\xi\mapsto\xi/2$ proves Eq.~\eqref{capuni:moment-form}.

For $d$ in Eq.~\eqref{capuni:classical-auxiliary}, set
$J=\int_0^1y^2(1-d^2y^2)^{-1}\,dy$.
The explicit parameter $\delta_{\mathrm C}$ obeys
\begin{align}
0<\frac{\delta_{\mathrm C}}a
&=2\gamma a\int_0^1
 \frac{y^2\,dy}{(1-d^2y^2)(1-\gamma^2y^2)}\notag\\
&\le\frac{2\gamma J}{1+\gamma},\notag\\
\frac{\kappa_{\mathrm{cl}}}{4a}
&=\frac1{1+\gamma}+\gamma J.
\label{capuni:classical-parameter-bound}
\end{align}
Since $bJ<1$, subtraction yields
\begin{equation}
\frac{\kappa_{\mathrm{cl}}}{4a}-\frac{\delta_{\mathrm C}}a
\ge\frac{1-bJ}{1+\gamma}>0.
\end{equation}
Thus $0<t_{\mathrm{cl}}=2\delta_{\mathrm C}/\kappa_{\mathrm{cl}}<1/2$.
The maximizing ensemble consists of the equiprobable signals
$\sqrt{1-u_{\mathrm{cl}}}\,|0\rangle
\pm\sqrt{u_{\mathrm{cl}}}\,|1\rangle$.

\paragraph{Entanglement-assisted capacity.}
Differentiating the mutual-information objective gives
\begin{equation}
-f_{\mathrm{ea}}''(u)
=\frac{a[2(1-u)+\gamma u^2]}
 {u(1-u)(1-au)(1-\gamma u)}>0.
\label{capuni:ea-curvature-function}
\end{equation}
With $c=c_{\mathrm{ea}}$, $r=1-c$, and
\begin{equation}
\alpha=\frac1r,\quad \beta=\frac{a}{1-ac},\quad
\chi=\frac{\gamma}{1-\gamma c},\quad \tau=\frac1c,
\label{capuni:ea-kernel-parameters}
\end{equation}
the logarithmic differences give
\begin{align}
K_{\mathrm{ea}}(w)
&=\int_0^\alpha\frac{d\xi}{1+\xi w}
 +a\int_0^\beta\frac{d\xi}{1+\xi w}\notag\\
&\quad-\gamma\int_0^\chi\frac{d\xi}{1+\xi w}
 +2a\int_0^\tau\frac{d\xi}{1-\xi w}.
\label{capuni:ea-raw-kernel}
\end{align}
The first three measures combine positively because $\chi\le\alpha$.
Moreover, $\beta,\chi,\tau\le\alpha=1/r$. Rescaling their supports
by $\xi\mapsto-r\xi$ and the last by $\xi\mapsto r\xi$, then
dividing by the total mass $\kappa_{\mathrm{ea}}$, proves
Eq.~\eqref{capuni:moment-form}.

To bound $t_{\mathrm{ea}}$, write
\begin{align}
T&=\frac{\gamma}{a}\ln\frac{1+ax}{\gamma},\notag\\
U&=\frac{f_{\mathrm{ea}}'(c)}a
 =L+\ln(1+\gamma x)-2\ln x-T,
\label{capuni:ea-U}\\
V&=\frac{\kappa_{\mathrm{ea}}r}{a}
 =\frac{1+\gamma+x}{1+ax}
  +\frac{a}{1+\gamma x}+\frac2x.
\label{capuni:ea-V}
\end{align}
Then $t_{\mathrm{ea}}=-U/V$. Let $L_0$ solve
$L_0-\ln(1+L_0)=1$; it satisfies $2<L_0<9/4$.
If $x=1$, then $\gamma<9/10$ and
$V=2(4+\gamma)/[(1+\gamma)(2-\gamma)]>3$.
The bounds $2\gamma<T<2$, together with $T\le\ln2$ for
$\gamma\le1/4$, give, by splitting at $\gamma=1/4$,
\begin{equation}
-\frac32<U<\ln19-\frac95<\frac32.
\end{equation}

If $x=L-\ln(1+L)>1$, then $L>L_0$ and the elementary
logarithm bounds $y/(1+y)\le\ln(1+y)\le y$ imply
\begin{equation}
0\le x+1-T\le e^{-L}(1+L+L^2)\le\frac7{e^2}<1.
\label{capuni:ea-T-bound}
\end{equation}
Also, $\gamma(1+L)>x$ and $(1+L)/x<13/4$, so
\begin{equation}
0<Y:=\ln\frac{(1+L)(1+\gamma x)}{x^2}
<\ln\frac{13}{2}<2.
\end{equation}
Since $U=Y+x-T$, this gives $-1<U<2$.
The function $ax=e^{-L}[L-\ln(1+L)]$ decreases for $L>L_0$,
so $ax<e^{-2}<1/7$ and $\gamma>6/7$. Consequently,
\begin{equation}
V>\frac{13+7x}{8}+\frac2x
\ge\frac{13}{8}+\sqrt7>4.
\end{equation}
Both branches therefore satisfy $|t_{\mathrm{ea}}|<1/2$.
The optimal input is the diagonal state $\rho_{u_{\mathrm{ea}}}$.
At $\gamma=1/2$, $f_{\mathrm{ea}}=H$, giving $C_E=1$;
the endpoint capacities follow from the identity and constant channels.
This completes the proofs of Theorems~\ref{capuni:thm-classical}
and~\ref{capuni:thm-ea}.

\subsection{Coefficients and quantum-capacity parameters}
\label{capuni:app-coefficients}

Writing $B_\nu(z)=1+\sum_{j\ge1}m_{\nu,j}z^j$, one has
\begin{equation}
m_{\nu,j}
=\frac{(-1)^{j+1}r_\nu^j f_\nu^{(j+2)}(c_\nu)}
       {\kappa_\nu(j+1)!},\qquad j\ge1.
\label{capuni:moments-derivatives}
\end{equation}
The coefficients in Eq.~\eqref{capuni:series-coefficients} are
the finite polynomials
\begin{equation}
A_{\nu,n}=\frac1n
\sum_{\boldsymbol\ell\in\mathcal L_n}
(-1)^{|\boldsymbol\ell|}(n)_{|\boldsymbol\ell|}
\prod_{j=1}^{n-1}\frac{m_{\nu,j}^{\ell_j}}{\ell_j!},
\label{capuni:finite-polynomials}
\end{equation}
where $\mathcal L_n$ consists of nonnegative integer tuples
satisfying $\sum_{j=1}^{n-1}j\ell_j=n-1$,
$|\boldsymbol\ell|=\sum_j\ell_j$, and
$(n)_k=\Gamma(n+k)/\Gamma(n)$ is the rising factorial.
The empty-tuple convention gives $A_{\nu,1}=1$;
the next coefficients are $A_{\nu,2}=-m_{\nu,1}$ and
$A_{\nu,3}=2m_{\nu,1}^2-m_{\nu,2}$.

For comparison with the quantum-capacity parameters, set
\begin{equation}
k_E=\frac{\gamma}{1-\gamma}=2\kappa_\gamma,\qquad
w_E=\frac{a u_{\mathrm{ea}}}{1-u_{\mathrm{ea}}}.
\end{equation}
Here $\kappa_\gamma$ is defined in
Eq.~\eqref{eq:effectiveparameters}. The EA stationarity condition is
\begin{equation}
w_E^2(1+w_E)^{k_E}=c_{k_E}(1+k_Ew_E),
\label{capuni:ea-quantum-root-comparison}
\end{equation}
whereas the quantum-capacity root at damping
$\gamma/(1+\gamma)$ satisfies
$(z_{k_E}-1)z_{k_E}^{k_E}=c_{k_E}$, with
$c_{k_E}$ from Eq.~\eqref{eq:kckQ}.
Thus the two problems share these parameters but have different
root equations; the classical and EA series instead share the
coefficient polynomials in Eq.~\eqref{capuni:finite-polynomials}.

\end{document}